\documentclass[11pt,a4paper]{article}

\usepackage{epsf,latexsym,amsfonts,amsbsy}
\usepackage{amsmath, amssymb, amsthm,amscd,amsxtra}

\makeatletter 

\@addtoreset{equation}{section}
\makeatother 

\newtheorem{theorem}{Theorem}[section]
\newtheorem{lemma}{Lemma}[section]
\newtheorem{remark}{Remark}[section]

\newtheorem{open question}{Open question}[section]

\newcommand{\ba}{\begin{array}}
\newcommand{\ea}{\end{array}}
\newcommand{\bea}{\begin{eqnarray}}
\newcommand{\eea}{\end{eqnarray}}

\newcommand{\eproof}{$\quad \Box$}

\newcommand{\reff}[1]{(\ref{#1})}

\begin{document}


\title{Improved Balas and Mazzola Linearization for Quadratic
$0$-$1$ Programs with Application in a New Cutting Plane Algorithm}
 \author{{Wajeb GHARIBI} \\
{\it\small Dept. of Computer Science,
 College of Computer Science $\&$ Information Systems, }\\
{\it\small   Jazan University, Jazan 82822-6694, KSA}\\
  {\it\small E-mail:  gharibi@jazanu.edu.sa }\\
   }

\date{}

\maketitle

\begin{abstract}
Balas and Mazzola linearization (BML) is widely used in devising
cutting plane algorithms for quadratic $0$-$1$ programs. In this
article, we improve BML by first strengthening the primal
formulation of BML and then considering the dual formulation.
Additionally, a new cutting plane algorithm is proposed.
\end{abstract}

{\bf Keywords:}\quad Quadratic Program, Integer program,\
Linearization,\ Cutting plane algorithm.


\section{Introduction}

In this article, we consider the generalized quadratic $0$-$1$
program given as follows
\begin{equation}\label{prob}\textnormal{(P)}\quad
\begin{array}{cc}
\min& x^TBx+c^Tx\\
s.t.&x\in X\subseteq \{0,1\}^n,\end{array}
\end{equation}
where $B=(b_{ij})$ is an $n\times n$ nonnegative matrix. Without
loss of generality we assume $b_{ii}=0$ since
$b_{ii}x_i^2=b_{ii}x_i$. Problem (P) is a generalization of
unconstrained zero-one quadratic problems, zero-one quadratic
knapsack problems, quadratic assignment problems and so on. It is a
classical NP-hard problem \cite{Garey}.

Linearization strategies are to reformulate the zero-one quadratic
programs as equivalent mixed-integer programming problems with
additional binary variables and/or continuous variables and
continuous constraints, see \cite{Chaovalitwongse,Elloumi,Fortet,
Glover1,Glover2,Glover3,Liberti}. Recently, Sherali and Smith
\cite{Sherali} developed small linearizations for more generalized
quadratic $0$-$1$ programs. Gueyea and Michelon \cite{Gueyea}
proposed a framework for unconstrained quadratic $0$-$1$ programs.
These linearizations are standard for employing exact algorithms
such as branch and bound.  Balas and Mazzola proposed a small-size
linearization \cite{Balas} and then successfully applied it to
devise exact or heuristic cutting plane algorithms.

In this article, we focus on  new  small-size tight linearizations.
We first propose a primal version of Balas and Mazzola linearization
(BML). By strengthening the linearization and then considering the
dual model, we obtain a new linearization which improves BML. As a
direct application, a new cutting plane algorithm is proposed.

This article is organized as follows. In section 2, we discuss Balas
and Mazzola linearization (BML) \cite{Balas} and the related primal
linearization. In section 3, we create a new approach to obtain a
tighter linearization. It improves the primal linearization of BML
in the sense that the linear programming relaxation often give
tighter lower bound. In section 4, we apply this dual linearization
to devise cutting plane algorithm and compare the efficacy with that
of BML. Concluding remarks are made in section 5.

\section{The Primal Model of Balas and Mazzola Linearization}
\setcounter{equation}{0}

In this section, we show that Balas and Mazzola Linearization has a
primal model.

Define a column vector $u$ with components
\begin{eqnarray}
 u_i = \max \{\sum_{j\neq i}b_{ij}x_j:\ x\in
 \overline{X}\},~i=1,\ldots,n,
\label{u:def}
\end{eqnarray}
where $\overline{X}$ is any suitable relaxation of $X$ such that the
problems \reff{u:def} can be solved relatively easily.

We rewrite the objective function of (P) as
$$
\sum_{i=1}^n ( x_i\sum_{j\neq i}b_{ij}x_j+ c_ix_i).
$$
Introducing $n$ continuous variables
\begin{equation}
y_i:=x_i\sum_{j\neq i}b_{ij}x_j,
\end{equation}
we can obtain the following mixed $0$-$1$ linear program
\begin{eqnarray}
(PL_1)\quad \begin{array}{cl} \min & \sum_{i=1}^n y_i+c_ix_i\\
s.t. & y_i\geq \sum_{j\neq i}b_{ij}x_j+ u_ix_i-u_i,\ i=1,...,n,\\
& y_i\geq 0,\ i=1,...,n,\\
& x\in X,\end{array}
\end{eqnarray}
where the two series inequality constraints follow from the fact
$(x_i-1)(\sum_{j\neq i}b_{ij}x_j-u_i)\geq 0$ and $x_i\geq 0$,
$b_{ij}\geq 0$, respectively.

\begin{theorem}\label{thm:1}
Problems ($P$) and ($PL_1$) are equivalent in the sense that for
each optimal solution to one problem, there exists an optimal
solution to the other problem having the same optimal objective
value.
\end{theorem}
The proof is found in Appendix \ref{s1}.

\begin{remark}
If we restrict (P) as the quadratic assignment problem, the proposed
linearization ($PL_1$) reduces to Kaufman-Broeckx linearization
\cite{Kaufman}.
\end{remark}

Below we apply Benders' decomposition approach to Problem (P), as in
\cite{Burkard}. Firstly, (P) can be decomposed in the following way
\begin{eqnarray}\label{decom}
\min_{x\in X}\{ \{ \min_{y\in Y(x)}\ \sum_{i=1}^n
y_{i}\}+\sum_{i=1}^n c_ix_i\},
\end{eqnarray}
where
\begin{eqnarray}
Y(x):=\{y\in \Re^{n}|\ y_{i}\geq \sum_{j\neq
i}b_{ij}x_{j}+u_{i}x_{i}-u_{i},\ y_{i}\geq 0,\ i=1,2,...n\} .
\end{eqnarray}
For fixed $x$, we dualize the first series constraints of the
problem $\min_{y\in Y(x)}\ \sum_{i=1}^ny_{i}$ using Lagrangian
multipliers $\lambda_{i}$ ($i=1,2,...n$). We obtain the subproblem
\begin{eqnarray}\label{spx}
SP(x):\ \begin{array}{cl}\max
&\sum_{i=1}^n(\sum_{j\neq i}b_{ij}x_{j}-u_{i}+u_{i}x_{i})\lambda_{i}\\
s.t.& 0\leq \lambda_{i}\leq 1 , \  i,j=1,2,...n.
\end{array}
\end{eqnarray}
Note that the feasible solution region $F$ of SP(x) does not depend
on the chosen vector $x$. Let $\lambda^{t}$ be the incidence vectors
of the extreme points of $F$ (which is unit hypercube in $\Re^{n}$).
Introducing
\begin{eqnarray}
&&\alpha^{(t)}_{i}:=\sum_{j\neq
i}\lambda^{(t)}_{j}b_{ji}+\lambda^{(t)}_{i}u_{i},\label{coeff1}\\
&&\beta^{(t)}:=\sum_{i=1}^n\lambda^{(t)}_{i}u_{i}\ ,\qquad
t=1,2,..., 2^{n}:=T ,\label{coeff2}
\end{eqnarray}
we can see that Problem (\ref{decom}) is equivalent to
\begin{eqnarray}
\min_{x\in X}\ \max_{1\leq t\leq T}\
\{\sum_{i=1}^n\alpha^{(t)}_{i}x_{i}-\beta^{(t)}\}+\sum_{i=1}^nc_{i}x_{i}
,\label{obj1}
\end{eqnarray}
by the fact that for any fixed $x$, the second-stage problem
$\min_{y\in Y(x)}\ \sum_{i=1}^ny_{i}$ of (\ref{decom}) is a linear
programming whose dual formulation is just (\ref{spx}) and the fact
that one of the optimal solutions to the linear programming problem
(\ref{spx}) is attained at an extreme point of $F$. Problem
(\ref{obj1}) yields now the following mixed $0$-$1$ linear program
\begin{eqnarray}\label{obj2}
(DL_1):\ \begin{array}{cl} \min & z+\sum_{i=1}^nc_{i}x_{i} \\
s.t.& z\geq\sum_{i=1}^n\alpha^{(t)}_{i}x_{i}-\beta^{(t)} ,\ 1\leq
t\leq T ,\\ &x\in X.\end{array}
\end{eqnarray}
In some sense, linearization ($DL_1$) can be regarded as the dual
formulation of (PL1). Above we also obtained the equivalence between
($PL_1$) and ($DL_1$):
\begin{theorem}\label{thm:2}
Problems ($PL_1$) and ($DL_1$) are equivalent in the sense that for
each optimal solution to one problem, there exists an optimal
solution to the other problem having the same optimal objective
value.
\end{theorem}
Combining Theorem \ref{thm:1} with Theorem \ref{thm:2}, we can see
($DL_1$) is equivalent to (P). In literature, linearization ($DL_1$)
is known as Balas and Mazzola linearization (BML) \cite{Balas}.

\section{New Tight Primal and Dual Linearizations}
\setcounter{equation}{0}

In this section, we propose a new approach to establish new tight
linearizations.

Define
\begin{eqnarray}
&&v_i = \max \{\sum_{j\neq i}b_{ij}x_j:\ x\in \overline{X},\
x_i=0\},~i=1,\ldots,n,
\label{v:def}\\
&&l_i = \min \{\sum_{j\neq i}b_{ij}x_j:\ x\in \overline{X},\
x_i=1\},~i=1,\ldots,n. \label{l:def}
\end{eqnarray}
Let $v$ and $l$ be the vectors with components $v_i$ and $l_i$
respectively, $i=1,\ldots,n$.
\begin{lemma}[\cite{Xia0}]\label{lem:1}
Let $x\in X\subseteq \{0,1\}^n$. For all $i=1,\ldots,n$,
\begin{eqnarray}
&&x_i\sum_{j\neq i}b_{ij}x_j = \max \{\sum_{j\neq i}b_{ij}x_j +v_i
x_i-v_i,\ l_i x_i \}. \label{rela:1:1}
\end{eqnarray}
\end{lemma}
Therefore, the new linearization reads
\begin{eqnarray}
(PL_2)\quad \begin{array}{cl} \min & \sum_{i=1}^n y_i+c_ix_i\\
s.t. & y_i\geq \sum_{j\neq i}b_{ij}x_j+ v_ix_i-v_i,\ i=1,...,n,\\
& y_i\geq l_ix_i,\ i=1,...,n.\\
& x\in X,\end{array}
\end{eqnarray}
Under the linear transformations $t_i:=y_i-l_ix_i,~i=1,\ldots,n$,
the above linearization becomes
\begin{eqnarray}
(PL_2')\quad \begin{array}{cl} \min & \sum_{i=1}^n t_i+(l_i+c_i)x_i\\
s.t. & t_i\geq \sum_{j\neq i}b_{ij}x_j+ (v_i-l_i)x_i-v_i,\ i=1,...,n,\\
& t_i\geq 0,\ i=1,...,n,\\
& x\in X.\end{array}
\end{eqnarray}
As a corollary of Lemma \ref{lem:1}, we have
\begin{theorem}\label{thm:3}
Problems ($PL_2$) (or ($PL_2'$)) and (P) are equivalent in the sense
that for each optimal solution to one problem, there exists an
optimal solution to the other problem having the same optimal
objective value.
\end{theorem}
Continuously relaxing linearizations ($PL_1$) and ($PL_2$), i.e.,
replacing $X$ with $\overline{X}$, we obtain linear programming
lower bounds for (P), denoted by $v(R\textnormal{-}PL1)$ and
$v(R\textnormal{-}PL_2)$ respectively. ($PL_2$) is not weaker than
($PL_1$) in the following sense.
\begin{theorem}\label{thm:4}
$v(R\textnormal{-}PL_1)\leq v(R\textnormal{-}PL_2)$.
\end{theorem}
\begin{proof}
It is sufficient to show any feasible solution of ($PL_2$) is also
feasible in ($PL_1$), which follows from the fact that $v_i\leq
u_i$, $x_i\in [0,1]$ and $l_i\geq 0$ since $B=(b_{ij})$ is
nonnegative.
\end{proof}

Next, we consider the dual model of ($PL_2$). As the formulation of
($PL_2'$) is similar to that of ($PL_1$), we immediately have the
dual model based on ($DL_2$).
\begin{eqnarray}\label{obj3}
(DL2):\ \begin{array}{cl} \min & z+\sum_{i=1}^n(l_i+c_{i})x_{i} \\
s.t.&
z\geq\sum_{i=1}^n\overline{\alpha}^{(t)}_{i}x_{i}-\overline{\beta}^{(t)}
,\ 1\leq t\leq T ,\\ &x\in X,\end{array}
\end{eqnarray}
where
\begin{eqnarray}
&&\overline{\alpha}^{(t)}_{i}:=\sum_{j\neq
i}\lambda^{(t)}_{j}b_{ji}+\lambda^{(t)}_{i}(v_{i}-l_i),\\
&&\overline{\beta}^{(t)}:=\sum_{i=1}^n\lambda^{(t)}_{i}v_{i}.
\end{eqnarray}
Similarly to Theorem \ref{thm:2}, we have
\begin{theorem}\label{thm:5}
Problems ($PL_2$) (or ($PL_2'$) ) and ($DL_2$) are equivalent in the
sense that for each optimal solution to one problem, there exists an
optimal solution to the other problem having the same optimal
objective value.
\end{theorem}

\section{Cutting Plane Algorithms Based on Dual Linearizations}
\setcounter{equation}{0}

We first establish cutting plane algorithm based on ($DL_1$). As in
any decomposition approach the master problem ($DL_1$) is not solved
for all restrictions
$z\geq\sum_{i=1}^n\alpha^{(t)}_{i}x_{i}-\beta^{(t)},\ (1\leq t\leq
T)$, but only for a subset $\{t|\ 1\leq t\leq r\}$ of indices. We
denote the restricted master problem by $({DL_1}^r)$.

Getting a solution $\widetilde{x}$ for the restricted master
problem, the subproblem $SP(\widetilde{x})$ is solved, which yields
\begin{eqnarray}
\lambda^{(r+1)}_{i}:=\widetilde{x}_{i} ,\  i,j=1,2,...n.
\end{eqnarray}
$\lambda^{(r+1)}$ is an optimal solution of the subproblem
$SP(\widetilde{x})$ because of the definition of the constants
$u_{i}$ and the constraints $0\leq \lambda_{i}\leq1$.

A new cut
\begin{eqnarray}\label{eq:0}
z\geq\sum_{i=1}^n\alpha^{(r+1)}_{i}x_{i}-\beta^{(r+1)}
\end{eqnarray}
is added to the current $({DL_1}^r)$. Thus we get $({DL_1}^{r+1})$.

The objection function value $\overline{z}$ of $SP(x)$ is an upper
bound for (P), whereas the objective function value $\underline{z}$
of the master problem $({DL_1}^r)$ is a lower bound. If
$\overline{z}=\underline{z}$, stop and return an optimal solution.

This is the flow of cutting plane algorithm. Below we show the
finite convergence. The proof is found in Appendix \ref{s2}.

\begin{lemma}\label{lem:2}
Assume that $\widetilde{x}$ is an optimal solution of $({DL_1}^r)$
and $\lambda^{(r+1)}_{i}:=\widetilde{x}_{i}$. For any $s>r$,
$\widetilde{x}$ cannot be an optimal solution of $({DL_1}^s)$ unless
it is the optimal solution of (P).
\end{lemma}

From the above lemma, we have the convergence result proved in
Appendix \ref{s3}.
\begin{theorem}
The above cutting plane algorithm stops in a finite number of steps
and returns the optimal solution of (P).
\end{theorem}

Cutting plane algorithm based on ($DL_2$) can be similarly devised.
To compare with ($DL_1$) conveniently, instead of ($DL_2$), we use
the following equivalent model
\begin{eqnarray}\label{obj4}
(DL_2'):\ \begin{array}{cl} \min & z+\sum_{i=1}^nc_{i}x_{i} \\
s.t.&
z\geq\sum_{i=1}^n(\overline{\alpha}^{(t)}_{i}+l_i)x_{i}-\overline{\beta}^{(t)}
,\ 1\leq t\leq T ,\\ &x\in X.\end{array}
\end{eqnarray}
\begin{theorem}\label{thm:6}
Assume $l_{i}> 0$ for all $i=1,2,...,n$. The restricted master
program $({{DL_2}^r}')$ gives a lower bound strictly better than
that of $({DL_1}^r)$ until the cutting plane algorithm stops.
\end{theorem}

\section{Conclusions}
In this article, we focus on the generalized quadratic $0$-$1$
program, denoted by (P). We propose a linearization ($PL_1$) for (P)
and show that it can be regarded as a dual formulation of Balas and
Mazzola linearization (BML), denoted ($DL_1$). By applying a new
approach, we establish a tight linearization ($PL_2$) of the same
size. We proved ($PL_2$) is not weaker than ($PL_1$) in the sense
that the continuous linear programming relaxation of ($PL_2$) gives
tighter lower bound than that of ($PL_1$).  The dual linearizations
of ($PL_1$) and ($PL_2$), ($DL_1$) and ($DL_2$) are successfully
used in devising cutting plane algorithms, respectively. We showed
that the cutting plane algorithm based on ($DL_2$) is strictly
better than that of ($DL_1$) under some weak assumptions.

\appendix

\section{$Proof~ of~ Theorem~ \ref{thm:1}:~$} \label{s1}
Let $x$ be any feasible solution to (P). It is easy to verify that
$(x,y)$ is feasible in ($PL_1$) with the same objective value, where
$y_i=x_i\sum_{j\neq i}b_{ij}x_j$. As a consequence, the optimal
objective value of ($PL_1$) gives a lower bound for (P). It is
sufficient to show that if $(x^*,y^*)$ is an optimal solution to
($PL_1$), $x^*$ is optimal in (P) with the same objective value. We
notice that $y_i^*=\max\{\sum_{j\neq i}b_{ij}x_j^*+ u_ix_i^*-u_i,0
\}$. If $x_i^*=1$, we have $y_i^*=\sum_{j\neq i}b_{ij}x_j^*$,
otherwise, $x_i^*=0$, $y_i^*=0$. As a conclusion,
$y_i^*=x_i^*(\sum_{j\neq i}b_{ij}x_j^*)$ which implies
$\sum_{i=1}^ny_i^*=\sum_{i=1}^nx_i^*(\sum_{j\neq i}b_{ij}x_j^*)$.
That is, $x^*$ is a feasible solution to (P) whose objective value
equals a lower bound. Therefore, $x^*$ is optimal in (P) and both
the optimal objective values are equal. \eproof

\section{$Proof~ of~ Lemma~ \ref{lem:2}:~$}\label{s2}
Denote the optimal objective function value of any master problem
$({DL_1}^s)$ by $\widetilde{z}_s$, which is a lower bound for (P).
If $\widetilde{x}$ is also an optimal solution of $({DL_1}^s)$ for
some $s>r$, it follows that
\begin{eqnarray}
\widetilde{z}_s\geq\sum_{i=1}^n\alpha^{(r+1)}_{i}\widetilde{x}_{i}-\beta^{(r+1)},\label{eq:1}
\end{eqnarray}
since $({DL_1}^s)$ contains the constraint (\ref{eq:0}). The
left-hand side of (\ref{eq:1}) is a lower bound for (P) while the
right-hand side of (\ref{eq:1}) corresponds to a feasible objective
function value of (P), which can be shown as follows:
\begin{eqnarray*}
&&\sum_{i=1}^n\alpha^{(r+1)}_{i}\widetilde{x}_{i}-\beta^{(r+1)}\\
&&=\sum_{i=1}^n(\sum_{j\neq
i}\lambda^{(r+1)}_{j}b_{ji}+\lambda^{(r+1)}_{i}u_{i})\widetilde{x}_{i}
-\sum_{i=1}^n\lambda^{(r+1)}_{i}u_{i}\\
&&=\sum_{i=1}^n(\sum_{j\neq
i}\widetilde{x}_{j}b_{ji}+\widetilde{x}_{i}u_{i})\widetilde{x}_{i}
-\sum_{i=1}^n\widetilde{x}_{i}u_{i}\\
&&=\sum_{i=1}^n(\sum_{j\neq
i}b_{ji}\widetilde{x}_{j})\widetilde{x}_{i}.
\end{eqnarray*}
Therefore (\ref{eq:1}) holds as equality and $\widetilde{x}$ must be
the optimal solution of (P). \eproof

\section{$Proof~ of~ Theorem~ \ref{thm:6}:~$}\label{s3}
If the cutting plane algorithm based on ($DL_2'$) has not stopped at
step $r$, the optimal solution $\widetilde{x}$ must be different
from $\lambda^{(t)}$ for any $1\leq t\leq r$, i.e., there exist
index $i_t$ such that
$(1-\lambda_{i_t}^{(t)})\widetilde{x}_{i_t}>0$. Then the right-hand
side of $({{DL_2}^r}')$ satisfies
\begin{eqnarray}
&&\sum_{i=1}^n(\overline{\alpha}^{(t)}_{i}+l_i)x_{i}-\overline{\beta}^{(t)}\nonumber\\
&&=\sum_{i=1}^n(\sum_{j\neq
i}\lambda^{(t)}_{j}b_{ji}+\lambda^{(t)}_{i}(v_{i}-l_i)+l_i)x_{i}-
\sum_{i=1}^n\lambda^{(t)}_{i}v_{i}\nonumber\\
&&\geq \sum_{i=1}^n(\sum_{j\neq
i}\lambda^{(t)}_{j}b_{ji}+\lambda^{(t)}_{i}u_{i})x_{i}-
\sum_{i=1}^n\lambda^{(t)}_{i}u_{i}+\sum_{i=1}^nl_i(1-\lambda_i^{(t)})x_i,\nonumber\\
&&> \sum_{i=1}^n(\sum_{j\neq
i}\lambda^{(t)}_{j}b_{ji}+\lambda^{(t)}_{i}u_{i})x_{i}-
\sum_{i=1}^n\lambda^{(t)}_{i}u_{i},\label{equl}
\end{eqnarray}
for any $1\leq t\leq r$. Therefore the objective function value of
$({{DL_2}^r}')$ is strictly larger than that of $({DL_1}^r)$.
\eproof



\begin{thebibliography}{99}
\bibitem{Al-Khayyal}A. Al-Khayyal, J.E. Falk: Jointly constrained biconvex
programming. \textit{Math.of Oper.Res. }\textbf{8(2)}: 273 (1983)
DOI: 10.1287/moor.8.2.273

\bibitem{Balas}E. Balas, J.B. Mazzola: Quadratic 0-1 programming by a
new linearization. Presented at the TIMS/ORSA meeting, Washington,
DC (1980)

\bibitem{Billionnet}A. Billionnet $\acute{E}$. Soutif: Using a Mixed Integer Programming Tool for Solving the 0-1
Quadratic Knapsack Problem. \textit{INFORMS Journal on Computing }
\textbf{16}: 188-197 (2004) DOI: 10.1287/ijoc.1030.0029

\bibitem{Burkard}R.E. Burkard, U. Derigs: Assignment and Matching
Problems: Solution Methods with FORTRAN-Programs. Springer-Verlag,
1980.

\bibitem{Chaovalitwongse} W. Chaovalitwongse, P.M. Pardalos, O.A. Prokopyev:
A new linearization technique for multi-quadratic $0$-$1$
programming problems. \textit{Operations Research Letters},
\textbf{32}: 517-522 (2004) DOI: 10.1016/j.orl.2004.03.005


\bibitem{Elloumi}
S. Elloumi, A. Faye, E. Soutif: Decomposition and Linearization for
0-1 Quadratic Programming. \textit{Annals of Operations Research},
\textbf{99}: 79-93 (2000) DOI: 10.1023/A:1019236832495


\bibitem{Fortet} R. Fortet: L'algebre de boole et ses applications en recherche operationnelle.
\textit{Cahiers du Centred'Etudes de Recheche Operationnelle},
\textbf{1}: 5-36 (1959) DOI: 10.1007/BF03006558


\bibitem{Garey}
M.R. Garey, D.S. Johnson: Computers and Intractability: A Guide to
the Theory of NP-Completeness. WH Freeman $\&$ Co, New York, NY,
USA, 1979

\bibitem{Glover1} F. Glover, Imporved linear integer programming formulations of nonlinear integer problems.
\textit{Manage. Sci.}, \textbf{22}(4): 455-460 (1975)
DOI:10.1287/mnsc.22.4.455

\bibitem{Glover2} F. Glover, E. Woolsey: Further reduction of zero-one polynomial programming problems
to zero-one linear programming problems. \textit{Oper. Res.},
\textbf{21}(1): 156-161 (1973) DOI: 10.1287/opre.21.1.156

\bibitem{Glover3} F. Glover, E. Woolsey: Converting the $0$-$1$ polynomial programming problem to a
$0$-$1$ linear program. \textit{Oper. Res.}, \textbf{22}(1): 180-182
(1974) DOI: 10.1287/opre.22.1.180

\bibitem{Gueyea}
S. Gueyea, P. Michelon: A linearization framework for unconstrained
quadratic (0-1) problems, \textit{Discrete Applied Mathematics},
\textbf{157}(6), 1255-1266 (2009) DOI: 10.1016/j.dam.2008.01.028



\bibitem{Kaufman}L. Kaufman, F. Broeckx: An algorithm for the quadratic assignment
problem using Bender's decomposition. \textit{European Journal of
Operational Research } \textbf{2}: 204-211 (1978) DOI:
10.1016/0377-2217(78)90095-4

\bibitem{Liberti}L. Liberti: Compact linearization of binary quadratic
problems, \textit{4OR},  \textbf{5}(3), 231-245, (2007) DOI:
10.1007/s10288-006-0015-3

\bibitem{McCormick}P. McCormick: Computability of global solution to
factorable nonconvex problems: Part I$-$Convex underestimating
problems. \textit{Math.Programming}\textbf{10}: 147-175 (1976) DOI:
10.1007/BF01580665


\bibitem{Sherali} H.D. Sherali, J.C. Smith: An improved linearization strategy for
zero-one quadratic programming problems. \textit{Optimization
Letters}, \textbf{1}: 33-47 (2007) DOI: 10.1007/s11590-006-0019-0

\bibitem{Xia0} W. Gharibi and Y. Xia: A tight linearization strategy for zero-one quadratic programming
problems. Submitted.

\bibitem{Xia} Y. Xia and Y. Yuan: A new linearization method for
quadratic assignment problems. \textit{Optimization Methods and
Software}, \textbf{21(5)}: 803-816 (2006)
DOI:10.1080/10556780500273077


\end{thebibliography}
\end{document}